\documentclass[10pt,conference]{IEEEtran}\IEEEoverridecommandlockouts
\usepackage{epsfig,graphicx,subfigure,psfrag,amsmath,cases}
\usepackage{latexsym,amssymb,epsfig,subfigure,algorithm}
\usepackage{algorithmic}
\usepackage{color}
\usepackage{url}
\usepackage{scrtime}
\usepackage{bm}

\author{\authorblockN{Shiyang Leng\authorrefmark{4}, Derrick Wing Kwan Ng\authorrefmark{1}, Nikola Zlatanov\authorrefmark{2}, and Robert Schober\authorrefmark{3}\thanks{Robert Schober is also with the University of British Columbia. This work was supported in part by the AvH Professorship Program of the Alexander von Humboldt Foundation.}}
The Pennsylvania State University, USA\authorrefmark{4}\\
The University of New South Wales, Australia\authorrefmark{1}\\
Monash University, Australia\authorrefmark{2}\\
Friedrich-Alexander-University Erlangen-N\"urnberg (FAU), Germany\authorrefmark{3}\\}

\title{Multi-Objective Resource Allocation in Full-Duplex SWIPT Systems}

\date{\thistime,\,\today}

\newtheorem{Thm}{Theorem}

\newtheorem{Prob}{Problem}
\newtheorem{T-Prob}{Transformed Problem}

\DeclareMathOperator{\Tr}{\mathrm{Tr}}
\DeclareMathOperator{\zero}{\mathbf{0}}
\DeclareMathOperator{\Rank}{\mathrm{Rank}}

\DeclareMathOperator{\diag}{\mathrm{diag}}

\DeclareMathOperator{\maxo}{\mathrm{maximize}}
\DeclareMathOperator{\mino}{\mathrm{minimize}}

 \newcommand{\qed}{\hfill \ensuremath{\blacksquare}}

\newcommand{\abs}[1]{\lvert#1\rvert}

\newcommand{\norm}[1]{\lVert#1\rVert}

\begin{document}

\maketitle

\begin{abstract}
In this paper, we investigate the resource allocation algorithm design for full-duplex simultaneous wireless information and power transfer (FD-SWIPT) systems. The considered system comprises a FD radio base station, multiple single-antenna half-duplex (HD) users, and multiple energy harvesters equipped with multiple antennas. We propose a multi-objective optimization framework to study the trade-off between uplink transmit power minimization, downlink transmit power minimization, and total harvested energy maximization. The considered optimization framework takes into account heterogeneous quality of service requirements for uplink and downlink communication and wireless power transfer. The non-convex multi-objective optimization problem is  transformed into an equivalent rank-constrained semidefinite program (SDP) and solved optimally by  SDP relaxation under certain general conditions. The solution of the proposed framework results in a set of Pareto optimal resource allocation policies.  Numerical results unveil an interesting trade-off between the considered
conflicting system design objectives and reveal the improved power efficiency facilitated by FD in SWIPT systems compared to traditional HD systems.
\end{abstract}
\renewcommand{\baselinestretch}{0.93}
\normalsize
\section{Introduction}
Next generation communication systems aim at
providing self-sustainability  and high data rates to communication networks with guaranteed
quality of service (QoS). A promising technique for prolonging the lifetime of communication networks is energy harvesting (EH). Among different EH technologies, wireless power transfer (WPT) via electromagnetic waves in radio frequency (RF) enables comparatively controllable EH at the receivers compared to conventional natural source, such as wind, solar, and tidal. Recent progress in the development of RF-EH circuitries has made RF-EH practical for  low-power consumption devices \cite{Krikidis2014}\nocite{JR:Kai_bin_magazine}--\cite{JR:Xiaoming_magazine}, e.g. wireless sensors. In particular, RF-EH enables  simultaneous wireless information and power transfer (SWIPT) \cite{CN:WIPT_fundamental}\nocite{JR:WIP_receiver,JR:non_linear_model}--\cite{JR:Rui_MISO_beamforming}. Thereby, as a carrier of both information and energy, the RF signal unifies information transmission and power transfer. Besides, RF-EH advocates energy saving by recycling the energy in the RF radiated by ambient transmitters. In SWIPT systems, the amount of harvested energy is an equally important QoS metric as the data rate and the transmit power consumption which are traditionally considered in communication networks. Thus, resource allocation algorithms for SWIPT systems should also take into account the emerging need for energy transfer \cite{JR:WIPT_fullpaper_OFDMA}--\nocite{JR:EE_SWIPT_Massive_MIMO}\cite{JR:MOOP}. In \cite{JR:WIPT_fullpaper_OFDMA}, energy-efficient SWIPT was investigated in multicarrier systems, where power allocation, user scheduling, and subcarrier allocation were considered.  In \cite{JR:EE_SWIPT_Massive_MIMO}, the authors proposed a power allocation scheme for  energy efficiency maximization of large scale multiple-antenna SWIPT systems. In \cite{JR:MOOP}, multi-objective optimization (MOO) was applied to jointly optimize multiple system design objectives to facilitate secure SWIPT systems. Although SWIPT has been already considered for various system setups, the power efficiency of  SWIPT systems \cite{JR:WIPT_fullpaper_OFDMA}--\nocite{JR:EE_SWIPT_Massive_MIMO}\cite{JR:MOOP}, has not been fully investigated and is still unsatisfactory due to the traditional half-duplex (HD) operation.

Recently, full-duplex (FD) communication has become a viable option for next generation wireless communication networks.
In contrast to  conventional HD transmission, FD communication  allows devices to transmit and receive simultaneously on the same frequency, thus potentially doubling the spectral efficiency. In practice,   the self-interference (SI)  caused by the own transmit signal impairs the simultaneous signal reception in FD systems severely which has been a major obstacle for the implementation of FD devices  in the past decades.  Fortunately,  breakthroughs in analog and digital self-interference cancellation (SIC) techniques \cite{FD_radios} have made FD communication more practical in recent years.  However, various  practical implementation issues, such as protocol and resource allocation algorithm design, need to be reinvestigated in the context of FD communications \cite{spec_eff_fd_smallcell}--\nocite{ofdma_fd_relay,CN:resall_fd_relay}\cite{CN:mulobj_fd_uldl}.
In \cite{spec_eff_fd_smallcell}, the authors proposed a suboptimal beamformer design to maximize the spectral efficiency of FD small cell wireless systems. In \cite{ofdma_fd_relay}, resource allocation and scheduling was studied for FD multiple-input multiple-output orthogonal frequency division multiple access (MIMO-OFDMA) relaying systems. Moreover, the energy efficiency of FD-OFDMA relaying systems was investigated in \cite{CN:resall_fd_relay}.  The authors of \cite{CN:mulobj_fd_uldl} proposed a multi-objective resource allocation algorithm for FD systems by considering the trade-off between uplink and downlink transmit power minimization. Although FD communication has  drawn  significant
research  interest \cite{spec_eff_fd_smallcell}--\nocite{ofdma_fd_relay,CN:resall_fd_relay}\cite{CN:mulobj_fd_uldl}, research on FD SWIPT systems is still in its infancy. Lately, the notion of FD  communication in EH
systems has  been pursued. Specifically, the combination of  FD and WPT was first considered in \cite{fd_wirlesspowered}. The authors optimized the resource allocation in a system with  WPT in the downlink and wireless information transmission in the uplink. In \cite{JR:Caijun_FD_SWIPT_relay}, the performance of a  dual-hop full-duplex relaying SWIPT system was studied. However,  simultaneous uplink and downlink communication has not been studied thoroughly for SWIPT systems. In fact, uplink and downlink transmission occurs simultaneously in FD systems and the associated information signals can also serve as vital energy sources for RF energy harvesting.  As a result, different trade-off naturally arise  in FD-SWIPT systems  when considering the aspects of uplink and downlink transmission as well as EH.  These observations motivate us to design a flexible resource allocation algorithm for FD SWIPT systems which strikes a balance between the different system design objectives.

The rest of the paper is organized as follows. In Section \ref{sect:OFDMA_AF_network_model},
we outline the system model for the considered FD SWIPT networks. In Section \ref{sect:forumlation}, we formulate the
multi-objective resource allocation algorithm design as a non-convex optimization
problem and solve this problem by  semidefinite programming relaxation. In Section \ref{sect:simulation}, we present numerical performance
results for the proposed optimal  algorithm. In Section \ref{sect:conclusion},
we conclude with a brief summary of our results.


\section{System Model}\label{sect:OFDMA_AF_network_model}
In this section, we first introduce the notation adopted in this paper. Then, we discuss the signal model for FD SWIPT networks.
\subsection{Notation}
We use boldface capital and lower case letters to denote matrices and vectors, respectively. $\mathbf{A}^H$, $\Tr(\mathbf{A})$, and $\Rank(\mathbf{A})$ represent the  Hermitian transpose, trace, and rank of  matrix $\mathbf{A}$, respectively;
$\diag(\mathbf{A})$ returns a diagonal matrix containing the diagonal elements of matrix $\mathbf{A
}$ on its main diagonal; $\mathbf{A}^{-1}$ and $\mathbf{A}^{\dagger}$ represent the inverse and Moore-Penrose pseudoinverse of matrix $\mathbf{A}$, respectively; $\mathbf{A}\succeq \mathbf{0}$ indicates that $\mathbf{A}$ is a positive semidefinite matrix; $\mathbf{I}_N$ is the $N\times N$ identity matrix; $\mathbb{C}^{N\times M}$ denotes the set of all $N\times M$ matrices with complex entries; $\mathbb{H}^N$ denotes the set of all $N\times N$ Hermitian matrices; $\abs{\cdot}$ and $\norm{\cdot}$ denote the absolute value of a complex scalar and the Euclidean vector norm, respectively; ${\cal E}\{\cdot\}$ denotes statistical expectation;  $[x]^+=\max\{x,0\}$; the circularly symmetric complex Gaussian distribution with mean vector $\boldsymbol{\mu}$ and covariance matrix $\boldsymbol{\Sigma}$ is denoted by ${\cal CN}(\boldsymbol{\mu},\boldsymbol{\Sigma})$; and $\sim$ stands for ``distributed as".
\begin{figure}[t]
\centering
\includegraphics[width=3.5in]{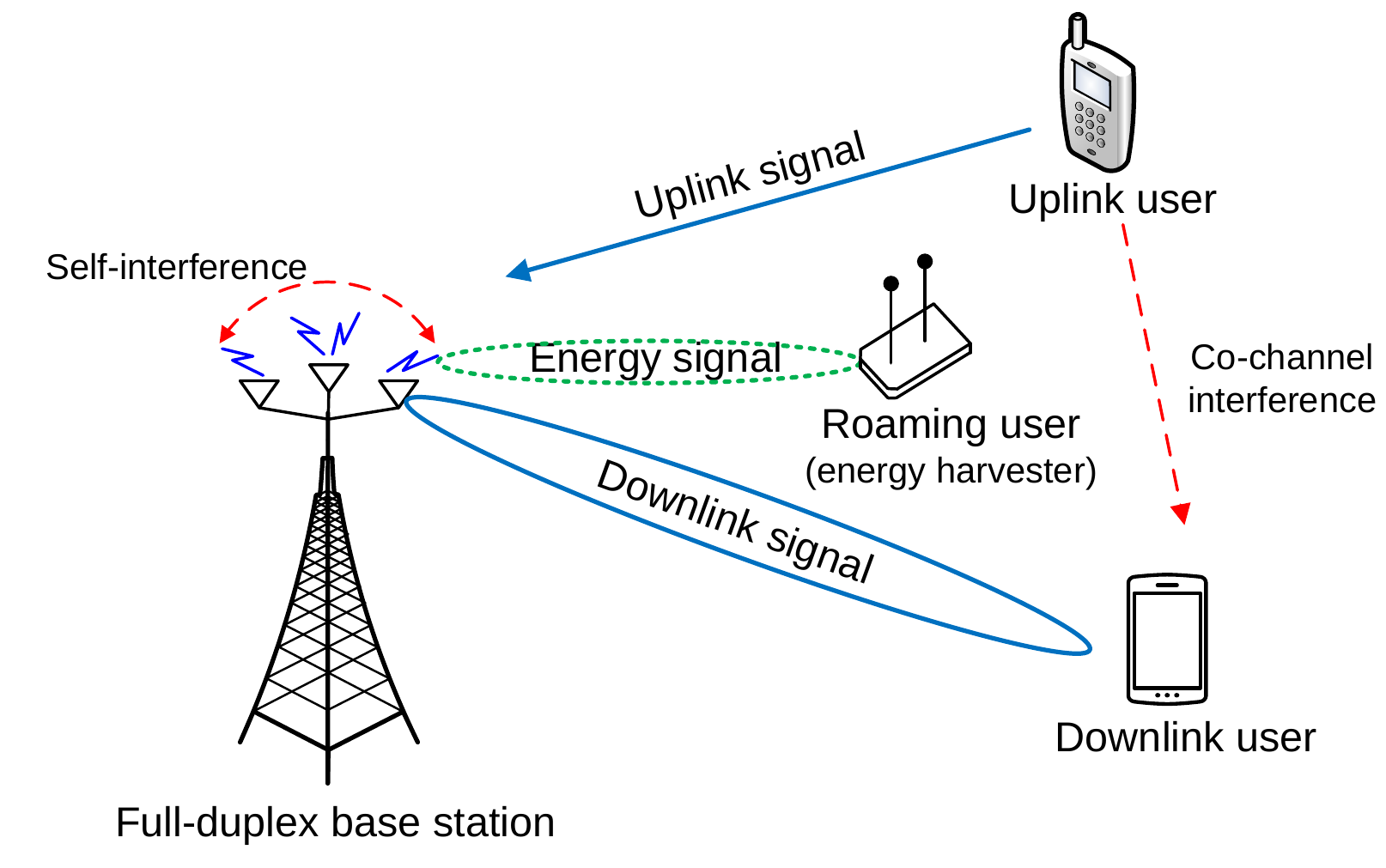}
\caption{Multiuser FD SWIPT system with a FD radio base station, $M=1$ uplink user, $K=1$ downlink user, and $J=1$ roaming user (energy harvester).}
\label{fig:system_model}
\end{figure}
\subsection{Signal Model}
We focus on a multiuser wireless communication system. The system consists of an FD radio base station (BS), $K$ HD downlink users, $M$ HD uplink users, and $J$ roaming users, cf. Figure \ref{fig:system_model}. The BS is equipped with $N>1$ antennas that can simultaneously perform downlink transmission and uplink reception in the same frequency band \cite{FD_radios}. All uplink and downlink users are single-antenna devices to limit the hardware complexity. On the other hand, to facilitate efficient EH, we assume that the roaming users are multiple-antenna devices, which are equipped with $N^\mathrm{EH}>1$ antennas.

 For downlink FD communication, $K$ independent signal streams are transmitted simultaneously at the same frequency to the $K$ downlink users. The transmitted signal at the FD radio BS is given by
\begin{eqnarray}\label{eqn:transig_BS}
\mathbf{x}=\sum_{k=1}^K\mathbf{w}_kd_k^\mathrm{DL}+\mathbf{q},
\end{eqnarray}
where $d_k^\mathrm{DL}\in\mathbb{C}$ is the information bearing signal intended for downlink user $k\in\{1,\ldots,K\}$. Without loss of generality, we assume ${\cal E}\{\abs{d_k^\mathrm{DL}}^2\}=1$. Besides, a beamforming vector $\mathbf{w}_k\in\mathbb{C}^{N\times1}$ is employed to assist downlink information transmission. On the other hand, in order to facilitate efficient WPT\footnote{ In this paper, we adopt a normalized unit energy, i.e., Joule-per-second. Thus, the terms ``energy" and ``power" are interchangeable.} to roaming users, a dedicated energy beam, $\mathbf{q}\in\mathbb{C}^{N\times1}$, is transmitted concurrently with the information signal. The energy signal $\mathbf{q}$ is modeled as a complex pseudo-random sequence with covariance matrix $\mathbf{Q}={\cal E}\{\mathbf{q}\mathbf{q}^H\}$. In general, both pseudo-random signals and constant amplitude signals are potential candidates for implementing the energy signal. However, pseudo-random energy signals can be shaped more easily to satisfy certain requirements on the spectrum mask of the transmit signal and are thus adopted in this paper. In particular, we assume that $\mathbf{q}$ is generated at the BS by a pseudo-random sequence generator with a predefined seed. This seed information is known at the downlink users. Thus, the interference caused by the energy signal can be completely cancelled at the downlink users before decoding the desired signals.

\subsection{Channel Model}
We consider a narrow-band slow fading channel. The received signal at downlink user $k$ is given by
\begin{eqnarray}\label{eqn:recesig_BS}
y_k^\mathrm{DL}=\mathbf{h}_k^H\mathbf{x}+\underbrace{\sum_{m=1}^M\sqrt{P_m}f_{m,k}d_m^\mathrm{UL}}_{\text{co-channel interference}}+n_k^\mathrm{DL},
\end{eqnarray}
where $\mathbf{h}_k\in\mathbb{C}^{N\times1}$ denotes the channel vector between the BS and downlink user $k$. The second term in (\ref{eqn:recesig_BS}) denotes the co-channel interference (CCI) caused by simultaneous uplink transmission in the FD system. $f_{m,k}\in\mathbb{C}$ is the channel gain between uplink user $m$ and downlink user $k$. $d_m^\mathrm{UL}$ and $P_m$ denote the uplink transmit signal from uplink user $m$ and the corresponding transmit power, respectively. We assume ${\cal E}\{\abs{d_m^\mathrm{UL}}^2\}=1$ without loss of generality. $n_k^\mathrm{DL}\sim{\cal CN}(0,\sigma_{\mathrm{DL},k}^2)$ denotes the additive white Gaussian noise (AWGN) at downlink user $k$.

At the same time, the FD BS receives signals from $M$ uplink users simultaneously. The corresponding received signal is given by
\begin{eqnarray}
\mathbf{y}^\mathrm{UL}&=&\sum_{m=1}^M\sqrt{P_m}\mathbf{g}_md_m^{\mathrm{UL}}\notag\\
&+&
\underbrace{\mathbf{c}}_{\text{self-interference cancellation noise}}+\mathbf{n}^\mathrm{UL},
\end{eqnarray}
where $\mathbf{g}_m\in\mathbb{C}^{N\times1}$ denotes the channel vector between uplink user $m$ and the BS.  Vector $\mathbf{n}^\mathrm{UL}\in\mathbb{C}^{N\times1}$ represents  the AWGN distributed as ${\cal CN}(\mathbf{0},\sigma_\mathrm{UL}^2\mathbf{I}_{N})$.  Due to the concurrent uplink reception and downlink  transmission at the FD radio BS, the SI caused by the downlink transmit signal impairs the uplink signal reception. In practice, different interference mitigation techniques  such as antenna cancellation, balun cancellation, and circulators \cite{CN:standford_FD,JR:FD_model} have been proposed to alleviate the impairment caused by SI. In order to isolate the resource allocation algorithm design from the specific implementation of self-interference mitigation, we model the self-interference cancellation induced noise by vector $\mathbf{c}\sim {\cal CN}(\zero,\varrho\diag({\cal E}\{\mathbf{H}_{\mathrm{SI}}(\mathbf{x}\mathbf{x}^H)\mathbf{H}_{\mathrm{SI}}^H\}))$ \cite[Eq. (4)]{JR:FD_model}, where
 $\mathbf{H}_{\mathrm{SI}}\in\mathbb{C}^{N\times N}$ is the self-interference channel and  $0\le\varrho\ll 1$ is a constant indicating the noisiness of the self-interference cancellation at the FD BS.

In the considered system, both downlink and uplink signals\footnote{In general, the adopted system model can be extended to scenarios in which the uplink users also transmit energy signal to facilitate EH. However, it may increase the peak-to-average power ratio and is not suitable for uplink users equipped with low cost power amplifiers.} act as energy sources to the roaming users (energy harvesters). The received signal at energy harvester $j\in\{1,\dots,J\}$ is
\begin{eqnarray}
\mathbf{y}_j^{\mathrm{EH}}=\bm{\Omega}_j^H\mathbf{x}+\sum_{m=1}^M\bm{\phi}_{j,m}\sqrt{P_m}d_m^\mathrm{UL}+\mathbf{n}_j^\mathrm{EH},
\end{eqnarray}
where matrix $\bm{\Omega}_j\in\mathbb{C}^{N\times N^\mathrm{EH}}$ and vector $\bm{\phi}_{j,m}\in\mathbb{C}^{N^\mathrm{EH}\times1}$ denote the channel between the BS and energy harvester $j$, and the channel between uplink user $m$ and energy harvester $j$, respectively. Vector $\mathbf{n}_j^\mathrm{EH}\in\mathbb{C}^{N^\mathrm{EH}\times1}$ represents the AWGN at energy harvester $j$ distributed as ${\cal CN}(\mathbf{0},\sigma_\mathrm{EH}^2\mathbf{I}_{N_\mathrm{R}})$.

We note that all channel variables, i.e., $\mathbf{h}_k$, $f_{m,k}$, $\mathbf{g}_m$, $\mathbf{H}_\mathrm{SI}$, $\bm{\Omega}_j$, and $\bm{\phi}_{j,m}$, capture the joint effect of path loss and small scale fading.

\section{Problem Formulation}\label{sect:forumlation}
In this section, we first introduce the adopted QoS metrics. Then, from the perspectives of uplink power consumption, downlink power consumption, and EH, we formulate three single objective optimization problems. In practice, these three system design objectives are all desirable but conflicting. Thus, we apply a MOO framework to study multi-objective resource allocation algorithm design.

\subsection{Quality of Service Metrics}
We assume that full channel state information (CSI) is available  at the FD BS for resource allocation. The receive signal-to-interference-plus-noise-ratio (SINR) at downlink user $k$ is given by
\begin{eqnarray}
\Gamma_k^\mathrm{DL}=\frac{\abs{\mathbf{h}_k^H\mathbf{w}_k}^2}{\displaystyle\sum_{i\neq k}^K\abs{\mathbf{h}_k^H\mathbf{w}_i}^2+\sum_{m=1}^MP_m\abs{f_{m,k}}^2+\sigma_{\mathrm{DL},k}^2},
\end{eqnarray}
where the interference from the energy beamforming signal, i.e., $\Tr(\mathbf{h}_k^H\mathbf{Q}\mathbf{h}_k)$, has already been cancelled since energy signal $\mathbf{q}$ is known to the downlink users.

For uplink transmission, we adopt zero-forcing beamforming (ZF-BF) for detection at the BS. In contrast to  optimal minimum mean square error beamforming (MMSE-BF) detection, ZF-BF facilitates the design of resource allocation algorithms in the considered problem. Additionally, the performance of ZF-BF converges to the performance of MMSE-BF in the  high SINR regime  \cite{book:wireless_comm}, which is the desired operating region of the considered system. Therefore, the receive SINR at the BS with respect to uplink user $m\in\{1,\dots,M\}$ can be expressed as
\begin{eqnarray}
\Gamma_m^\mathrm{UL}\hspace*{-1mm}=\hspace*{-1mm}\frac{P_m\abs{\mathbf{g}_m^H\mathbf{z}_m}^2}{\displaystyle\sum_{i\neq m}^MP_i\abs{\mathbf{g}_i^H\mathbf{z}_m}^2\hspace*{-1mm}+S^{\mathrm{UL}}_m
+\sigma_\mathrm{UL}^2\norm{\mathbf{z}_m}^2},
\end{eqnarray}
where
\begin{eqnarray}
S^{\mathrm{UL}}_m=\varrho\mathbf{z}_m^H\Big(\diag\big(\mathbf{H}_\mathrm{SI}\big(\sum_{k=1}^K\mathbf{w}_k\mathbf{w}_k^H\hspace*{-1mm}
+\mathbf{Q}\big)\mathbf{H}_\mathrm{SI}^H\big)\Big)\mathbf{z}_m
\end{eqnarray}
is the noise caused by SI cancellation and $\mathbf{z}_m\in\mathbb{C}^{N\times1}$ denotes the ZF-BF receive vector for decoding the signal of uplink user $m$. The ZF-BF matrix is given by
\begin{eqnarray}
\mathbf{Z}&=&[\mathbf{z}_1,\dots,\mathbf{z}_M]^T=(\mathbf{G}^H\mathbf{G})^{-1}\mathbf{G}^H,\\
\text{where}\quad\mathbf{G}&=&[\mathbf{g}_1,\dots,\mathbf{g}_M].\notag
\end{eqnarray}

On the other hand, the total amount of harvested energy at energy harvester $j\in\{1,\dots,J\}$ is given by
\begin{eqnarray}\label{eqn:P_harv}
P_j^\mathrm{EH}\hspace*{-1mm}=\hspace*{-1mm}\eta_j\bigg[\hspace*{-0.5mm}\Tr\hspace*{-0.5mm}
\Big(\hspace*{-0.5mm}\bm{\Omega}_j^H(\hspace*{-0.5mm}\sum_{k=1}^K\hspace*{-0.5mm}\mathbf{w}_k\mathbf{w}_k^H\hspace*{-0.5mm}+\hspace*{-0.5mm}\mathbf{Q})\bm{\Omega}_j\hspace*{-0.5mm}\Big)
\hspace*{-1mm}+\hspace*{-2mm}\sum_{m=1}^M\hspace*{-1mm}P_m\norm{\bm{\phi}_{j,m}}^2\hspace*{-0.5mm}\bigg],
\end{eqnarray}
where $0\leq\eta_j\leq1$ is the energy conversion efficiency of energy harvester $j$. It represents the energy loss in converting the received RF energy to electrical energy for storage. Note that the thermal noise power  is ignored in (\ref{eqn:P_harv}) for EH as it is negligibly  small compared to the power of the received signals.

\subsection{Optimization Problem Formulation}
In  FD SWIPT systems, downlink transmit power minimization, uplink transmit power minimization, and total harvested energy maximization are all desirable system design objectives. Now,  we first propose three single-objective optimization problems with respect to these objectives.
\vspace*{1mm}
\begin{Prob}\emph{Downlink Transmit Power Minimization:}\label{prob:dl_mino}
\vspace*{-1mm}
\begin{eqnarray}\label{eqn:dl_mino}
&&\underset{\mathbf{Q}\in\mathbb{H}^{N},\underline{\mathbf{w}},{\cal P}}{\mino}\,\, \sum_{k=1}^K\norm{\mathbf{w}_k}^2+\Tr(\mathbf{Q})\notag\\
\mathrm{s.t.}&&\mathrm{C1}:\,\,\sum_{k=1}^K\norm{\mathbf{w}_k}^2+\Tr(\mathbf{Q})\leq P_{\mathrm{max}}^\mathrm{DL},\notag\\
&&\mathrm{C2}:\,\,P_m\leq P_{\mathrm{max},m}^\mathrm{UL},\,\,\forall m,\notag\\
&&\mathrm{C3}:\,\,\Gamma_k^\mathrm{DL}\ge\Gamma_{\mathrm{req},k}^\mathrm{DL},\,\,\forall k,\notag\\
&&\mathrm{C4}:\,\,\Gamma_m^\mathrm{UL}\ge\Gamma_{\mathrm{req},m}^\mathrm{UL},\,\,\forall m,\notag\\
&&\mathrm{C5}:\,\,P_j^\mathrm{EH}\ge P_{\mathrm{min},j},\,\,\forall j,\notag\\
&&\mathrm{C6}:\,\,P_m\ge0,\,\,\forall m,\quad\mathrm{C7}:\,\,\mathbf{Q}\succeq\mathbf{0},
\end{eqnarray}
\end{Prob}
where $\underline{\mathbf{w}}=\{\mathbf{w}_k,\forall k\}$ and ${\cal P}=\{P_m,\forall m\}$ denote the downlink beamforming vector policy and the uplink transmit power policy, respectively.
In  \eqref{eqn:dl_mino}, we minimize the total downlink transmit power by jointly optimizing downlink information beamforming vectors $\mathbf{w}_k, \forall k$, the covariance matrix of energy signal, $\mathbf{Q}$, and uplink transmit power $P_m,\forall m$. Constants $P_{\mathrm{max}}^\mathrm{DL}$ and $P_{\mathrm{max},m}^\mathrm{UL}$ in C1 and C2 denote the maximum downlink transmit power for the FD BS and the maximum transmit power of uplink user $m$, respectively. QoS requirements of reliable communication are taken into account in  C3 and C4. In particular, $\Gamma_{\mathrm{req},k}^\mathrm{DL}>0, \forall k,$ and $\Gamma_{\mathrm{req},m}^\mathrm{UL}>0, \forall m,$ are the minimum required SINRs for the downlink and uplink users, respectively. $P_{\mathrm{min},j}, \forall j,$ in C5 is the minimum required amount of harvested energy for energy harvester $j$. In addition, C6 and C7 enforce the non-negative uplink transmit power constraints and the positive semidefinite Hermitian matrix constraint for covariance matrix $\mathbf{Q}$, respectively.

On the other hand, for the system designs with the objectives of uplink transmit power minimization and total harvested energy maximization, respectively, we have the same constraint set as for Problem 1. Therefore, the problem formulations for these two other system design objectives are given as, respectively,
\vspace*{1mm}
\begin{Prob}\emph{Uplink Transmit Power Minimization:}\label{prob:ul_mino}
\vspace*{-1mm}
\begin{eqnarray}
\underset{\mathbf{Q}\in\mathbb{H}^{N},\underline{\mathbf{w}},{\cal P}}{\mino}\,\,\sum_{m=1}^MP_m\notag\\
\mathrm{s.t.}\quad\mathrm{C1}-\mathrm{C7},
\end{eqnarray}
\end{Prob}
\begin{Prob}\emph{Total Harvested Energy Maximization:}\label{prob:eh_maxo}
\vspace*{-1mm}
\begin{eqnarray}
\underset{\mathbf{Q}\in\mathbb{H}^{N},\underline{\mathbf{w}},{\cal P}}{\maxo}\,\,\sum_{j=1}^JP_j^\mathrm{EH}\notag\\
\mathrm{s.t.}\quad\mathrm{C1}-\mathrm{C7}.
\end{eqnarray}
\end{Prob}
The interdependency between the aforementioned objectives is non-trivial in the considered FD SWIPT system. For instance,   although a large transmit power ensures high received SINRs at the downlink users,  the strong SI impairs the reception of the uplink signals at the FD BS. Similarly, increasing the uplink transmit power to satisfy a more stringent uplink SINR requirement will lead to severe CCI which degrades the downlink signal reception. On the other hand, the EH QoS requirement has to be fulfilled by transferring a sufficient amount of power in both uplink and downlink. Yet, minimizing either uplink or downlink transmit power conflicts with the objective of having a higher power for EH. Hence, a non-trivial trade-off between these three system design objectives naturally arises in the considered FD SWIPT system. Thus, a flexible resource allocation algorithm which can accommodate diverse system design preferences is desired. To this end, we apply MOO to  systematically address this resource allocation problem.

In the literature, MOO is commonly adopted as a mathematical framework to study the trade-off between multiple desirable but conflicting system design objectives. The optimal solution of a MOO program (MOOP) is defined by a Pareto optimal set; a set of points that satisfy the concept of Pareto optimality \cite{JR:MOOP}. In the following, we formulate a MOOP based on the weighted Tchebycheff method \cite{JR:MOOP}, in which the preferences for the aforementioned single system design objectives are quantified by a set of pre-specified weights. In fact, compared to other approaches to formulate MOOPs, the weighted Tchebycheff method can provide a complete Pareto optimal set by varying the weights, even if the MOOP is non-convex. For the sake of notational simplicity, we denote the objective functions of Problems 1--3 as $F_n(\mathbf{Q},\underline{\mathbf{w}},{\cal P})$, $n\in\{1,2,3\}$, respectively. Then, the MOOP is given by
\vspace*{1mm}
\begin{Prob}\emph{Multi-Objective Optimization:}\label{prob:mul_obj}\vspace*{-1mm}
\begin{eqnarray}
\underset{\mathbf{Q}\in\mathbb{H}^{N},\underline{\mathbf{w}},{\cal P}}{\mino}&& \max_{n=1,2,3}\,\, \Big\{\lambda_n\Big(F_n(\mathbf{Q},\underline{\mathbf{w}},{\cal P})-F_n^*\Big)\Big\}\notag\\
\mathrm{s.t.}&&\mathrm{C1}-\mathrm{C7},
\end{eqnarray}
\end{Prob}
where $F_n^*$ is the optimal objective value with respect to Problem $n\in\{1,2,3\}$. In order to represent the three single system design objective functions in a unified manner, without loss of generality, the maximization in {Problem \ref{prob:eh_maxo}} was rewritten as an equivalent minimization. As a result, $F_3(\mathbf{Q},\underline{\mathbf{w}},{\cal P})$ in Problem \ref{prob:mul_obj} is given by $F_3(\mathbf{Q},\underline{\mathbf{w}},{\cal P})=-\sum_{j=1}^JP_j^\mathrm{EH}$. Constant  $\lambda_n$ is a weight imposed on the $n$-th objective function subject to $0\leq\lambda_n\leq1$ and $\sum_n\lambda_n=1$, which indicates the preference of the system designer for the $n$-th objective function over the others. We can obtain a set of resource allocation policies by solving {Problem \ref{prob:mul_obj}} for different predefined weights. In the extreme case, when $\lambda_n=1$ and $\lambda_l=0, \forall l\neq n $, {Problem \ref{prob:mul_obj}} is equivalent\footnote{Here, equivalent means that both problems have the same solution.} to the $n$-th single-objective optimization problem.
\subsection{Optimal Solution}
Problems \ref{prob:dl_mino}-\ref{prob:mul_obj}  are non-convex optimization problems due to the non-convex constraints C3 and C4. To overcome the non-convexity, we recast these problems as SDPs via SDP relaxation. To this end, we define  new variables
\begin{eqnarray}
\mathbf{W}_k&=&\mathbf{w}_k\mathbf{w}_k^H, \mathbf{H}_k\hspace*{-0.5mm}=\hspace*{-0.5mm}\mathbf{h}_k\mathbf{h}_k^H,\,\,
\mathbf{G}_m\hspace*{-0.5mm}=\hspace*{-0.5mm}\mathbf{g}_m\mathbf{g}_m^H,\,\,\\
\mathbf{Z}_m\hspace*{-0.5mm}&=&\hspace*{-0.5mm}\mathbf{z}_m\mathbf{z}_m^H,\,\, \mbox{ and }
\bm{\Phi}_{j,m}\hspace*{-0.5mm}=\hspace*{-0.5mm}\bm{\phi}_{j,m}\bm{\phi}_{j,m}^H.
\end{eqnarray}
Thus, the considered problems can be equivalently transformed as follows:
\vspace*{2mm}
\begin{T-Prob}\label{prob:dl_mino_tf}
\end{T-Prob}\vspace*{-2mm}
\begin{eqnarray}
&&\underset{\underline{\mathbf{W}},\mathbf{Q}\in\mathbb{H}^{N},{\cal P}}{\mino}\,\, \Tr\Big(\sum_{k=1}^K\mathbf{W}_k+\mathbf{Q}\Big)\notag\\
\mathrm{s.t.}&&\hspace*{-5mm}\mathrm{C2, C6, C7,}\notag\\
&&\hspace*{-5mm}\overline{\mathrm{C1}}:\Tr\big(\sum_{k=1}^K\mathbf{W}_k+\mathbf{Q}\big)\leq P_{\mathrm{max}},\notag\\
&&\hspace*{-5mm}\overline{\mathrm{C3}}:\frac{\Tr(\mathbf{H}_k\mathbf{W}_k)}{\Gamma_{\mathrm{req},k}^\mathrm{DL}}\ge I_k^\mathrm{DL}+\sigma_{\mathrm{DL},k}^2,\,\,\forall k,\notag\\
&&\hspace*{-5mm}\overline{\mathrm{C4}}:\frac{P_m\Tr\big(\mathbf{G}_m\mathbf{Z}_m\big)}{\Gamma_{\mathrm{req},m}^\mathrm{UL}}\ge I_m^\mathrm{UL}+\sigma_\mathrm{UL}^2\Tr(\mathbf{Z}_m),\,\,\forall m,\notag\\
&&\hspace*{-5mm}\overline{\mathrm{C5}}:\overline{P_j^\mathrm{EH}}\ge P_{\mathrm{min},j},\,\,\forall j,\notag\\
&&\hspace*{-5mm}\mathrm{C8}:\mathbf{W}_k\succeq\mathbf{0},\,\,\forall k,\quad\mathrm{C9}:\mathrm{Rank}(\mathbf{W}_k)\leq1,\,\,\forall k,
\end{eqnarray}
where
\begin{eqnarray}
\hspace*{-12mm}I_k^\mathrm{DL}\hspace*{-2mm}&=&\hspace*{-2mm}\sum_{i\neq k}^K\Tr(\mathbf{H}_k\mathbf{W}_i)+\sum_{m=1}^M P_m\abs{f_{m,k}}^2,\\
\hspace*{-12mm}I_m^\mathrm{UL}\hspace*{-2mm}&=&\hspace*{-2mm}\sum_{i\neq m}^M\hspace*{-1mm}P_i\hspace*{-0.5mm}\Tr(\hspace*{-0.5mm}\mathbf{G}_i\mathbf{Z}_m\hspace*{-0.5mm})\notag\\
&+&\hspace*{-1mm}\varrho\Tr\Big(\mathbf{Z}_m\diag\big(\mathbf{H}_\mathrm{SI}\big(\sum_{k=1}^K\mathbf{W}_k
+\mathbf{Q}\big)\mathbf{H}_\mathrm{SI}^H\big)\hspace*{-0.8mm}\Big),\\
\hspace*{-12mm}\overline{P_j^\mathrm{EH}}\hspace*{-2mm}&=&\hspace*{-2mm}\eta_j\Big[\hspace*{-0.5mm}\Tr\hspace*{-0.5mm}\Big(\hspace*{-0.5mm}\bm{\Omega}_j^H\hspace*{-0.5mm}\big(\hspace*{-0.5mm}\sum_{i=k}^K\hspace*{-0.5mm}\mathbf{W}_k\hspace*{-1mm}
+\hspace*{-1mm}\mathbf{Q}\big)\bm{\Omega}_j\hspace*{-1mm}\Big)\hspace*{-1mm}
+\hspace*{-2mm}\sum_{m=1}^M\hspace*{-1mm}P_m\hspace*{-1mm}\Tr(\bm{\Phi}_{j,m}\hspace*{-0.5mm})\Big],
\end{eqnarray} and $\underline{\mathbf{W}}=\{\mathbf{W}_k,\forall k\}$ is the set of downlink beamforming matrices to be optimized. Constraints C8, C9, and $\mathbf{W}_k\in\mathbb{H}^{N}$ are introduced due to the definition of  $\mathbf{W}_k$. Similarly, Problems 2-4 are equivalently transformed to

\begin{T-Prob}\label{prob:ul_mino_tf}\end{T-Prob}
\vspace*{-3mm}
\begin{eqnarray}
\underset{\underline{\mathbf{W}},\mathbf{Q}\in\mathbb{H}^{N},{\cal P}}{\mino}\,\,\sum_{m=1}^MP_m\notag\\
\mathrm{s.t.}\quad\overline{\mathrm{C1}}-\mathrm{C9}.
\end{eqnarray}

\begin{T-Prob}\label{prob:eh_maxo_tf}\end{T-Prob}
\vspace*{-3mm}
\begin{eqnarray}
\underset{\underline{\mathbf{W}},\mathbf{Q}\in\mathbb{H}^{N},{\cal P}}{\maxo}\,\,\sum_{j=1}^J\overline{P_j^\mathrm{EH}}\notag\\
\mathrm{s.t.}\quad\overline{\mathrm{C1}}-\mathrm{C9}.
\end{eqnarray}

\begin{T-Prob}\label{prob:mul_obj_tf}
\vspace*{-1mm}
\begin{eqnarray}\label{eqn:mul_obj_tf}
&&\hspace*{-12mm}\underset{\underline{\mathbf{W}},\mathbf{Q}\in\mathbb{H}^{N},{\cal P},\tau}{\maxo}\quad\tau\notag\\
\mathrm{s.t.}\hspace*{-1mm}&&\hspace*{-4mm}\overline{\mathrm{C1}}-\mathrm{C9},\notag\\
&&\hspace*{-4mm}\mathrm{C10}:\lambda_n\Big(F_n(\mathbf{Q},\underline{\mathbf{W}},{\cal P})\hspace*{-1mm}-\hspace*{-1mm}F_n^*\Big)\leq\tau,\,\,n\hspace*{-1mm}\in\hspace*{-1mm}\{1,\hspace*{-0.5mm}2,\hspace*{-0.5mm}3\},
\end{eqnarray}
\end{T-Prob}
where $\tau$ is an an auxiliary optimization variable \cite{book:convex}.

Evidently, {Transformed Problem \ref{prob:mul_obj_tf}} is a generalization of {Transformed Problems \ref{prob:dl_mino_tf}-\ref{prob:eh_maxo_tf}}. Hence, we focus on the methodology for solving {Transformed Problem \ref{prob:mul_obj_tf}} in the following. {Transformed Problem \ref{prob:mul_obj_tf}} is non-convex due to the rank-one matrix constraint C9. To obtain a tractable problem formulation, we apply SDP relaxation. Specifically, we relax constraint C9 in (\ref{eqn:mul_obj_tf}) by removing it from the problem. Then, the considered problem becomes
\begin{eqnarray}\label{eqn:mul_obj_sdp}
&&\hspace*{-12mm}\underset{\underline{\mathbf{W}},\mathbf{Q}\in\mathbb{H}^{N},{\cal P},\tau}{\maxo}\quad\tau\notag\\
\mathrm{s.t.}&&\hspace*{-4mm}\overline{\mathrm{C1}}-\mathrm{C8},\notag\\
&&\hspace*{-4mm}\mathrm{C10}:\lambda_n\Big(F_n(\mathbf{Q},\underline{\mathbf{W}},{\cal P})\hspace*{-1mm}-\hspace*{-1mm}F_n^*\Big)\leq\tau,\,\,n\hspace*{-1mm}\in\hspace*{-1mm}\{1,\hspace*{-0.5mm}2,\hspace*{-0.5mm}3\}.
\end{eqnarray}
We note that the rank constraint relaxed problem in (\ref{eqn:mul_obj_sdp}) is a  convex SDP which can be solved by standard numerical convex program solvers such as CVX \cite{website:CVX}. In particular, if the obtained solution of the relaxed problem satisfies the rank-one constraint, i.e., $\Rank(\mathbf{W}_k^*)\leq 1$, then the solution of \eqref{eqn:mul_obj_sdp} is the optimal solution of the original {Problem \ref{prob:mul_obj}}. Thus, the optimal beamforming vector $\mathbf{w}_k^*$ of the original problem can be retrieved by solving the relaxed problem. Now, we reveal the tightness of the SDP relaxation by the following theorem.
\begin{Thm}\label{thm:rankone}
Assuming that the channels $\bm{\Omega}_j$, $\mathbf{H}_{\mathrm{SI}}$, and $\mathbf{h}_k$, are  statistically independent,
the optimal beamforming matrix for (\ref{eqn:mul_obj_sdp}) is a rank-one matrix, i.e.,  $\Rank(\mathbf{W}_k^*)=1,\forall k$, and the energy beamforming matrix satisfies $\Rank(\mathbf{Q}^*)\leq1$ with probability one for $\Gamma_{\mathrm{req}_k}^{\mathrm{DL}}>0$.
\end{Thm}
\begin{proof}
Please refer to the Appendix.
\end{proof}

In other words, whenever the  channels satisfy  the general condition
stated in Theorem
\ref{thm:rankone}, the adopted SDP relaxation is tight. Hence, the optimal solution of the original MOOP can be obtained by solving the relaxed SDP problem in (\ref{eqn:mul_obj_sdp}). Besides, information beamforming, i.e., $\Rank(\mathbf{W}_k^*)=1$, and energy beamforming, i.e., $\Rank(\mathbf{Q}^*)\leq1$, is optimal for optimizing the considered conflicting objective functions. On the other hand, the optimal solutions of the single-objective problems can be achieved by solving special cases of (\ref{eqn:mul_obj_sdp}). For instance,  the solution of single-objective {Problem \ref{prob:dl_mino}} can be obtained by solving (\ref{eqn:mul_obj_sdp}) with  $\lambda_1=1$, $\lambda_2=0$, and $\lambda_3=0$.

\section{Results}\label{sect:simulation}
In this section, we investigate the performance of the proposed multi-objective resource allocation algorithm. The important simulation parameters are summarized in Table \ref{table:parameters}. We evaluate a system with an FD radio BS located at the center of a cell. Furthermore, $K=3$ downlink users and $M=8$ uplink users located in the range between the reference distance of $10$ meters and the maximum distance of $50$ meters. $J=2$ energy harvesters are located close to the FD BS at a distance of between $2$ to $10$ meters in order to facilitate EH. Each energy harvester is equipped with $N^\mathrm{EH}=3$ antennas. The small scale fading of the uplink and downlink channels is modeled as independent and identically distributed Rayleigh fading. The EH channel and the SI channel are modeled as Rician fading channels with Rician factor $6$ dB. The maximum transmit power supply in downlink and uplink are  $P_\mathrm{max}^\mathrm{DL}=46$ dBm and $P_{\mathrm{max},m}^\mathrm{UL}=23$ dBm, $\forall m$, respectively. Without loss of generality, we assume that the required SINRs at all downlink users are identical. Besides, we specify $\Gamma_{\mathrm{req},m}^\mathrm{UL}=15$ dBm, $\forall m$, for uplink users. At the energy harvesters, the minimum required harvested energy is $P_{\mathrm{min},j}=-20$ dBm, $\forall j$.
\begin{table}[t]
\centering
\caption{Simulation Parameters.} \label{table:parameters}
\begin{tabular}{ | l | l | } \hline
      Carrier center frequency                           & 915 MHz\\ \hline
      Bandwidth                                          & 200 kHz \\ \hline 
      Antenna gain at FD BS                                & 10 dBi \\ \hline
            Antennas gain at users                               & 0 dBi \\ \hline
      Downlink user noise power                          & -71 dBm \\ \hline
      BS noise power                            & -83 dBm \\ \hline
      SI cancellation constant $\varrho$                                  & -110 dB \\ \hline
      Energy conversion efficiency, $\eta_j$                      & 0.8 \\ \hline
\end{tabular}
\end{table}

\subsection{Trade-off Region of Multiple Design Objectives}
Figure \ref{fig:trade_off_3D} depicts the trade-off region for uplink transmit power minimization, downlink transmit power minimization, and total harvested energy maximization achieved by the proposed optimal scheme.  There are $N=8$ transmit antennas at the BS and the minimum required downlink SINR is $\Gamma_{\mathrm{req},k}^\mathrm{DL}=21$ dBm, $\forall k$. The points shown for the trade-off region were obtained by solving the SDP relaxed problem for different sets of weights $0\leq\lambda_n\leq1, n=1,2,3$ subject to $\sum_n\lambda_n=1$. As can be seen, there is a nontrivial trade-off between uplink and downlink transmit power minimization and total harvested energy maximization. In particular, for a fixed weight $\lambda_3$ for EH maximization, the downlink transmit power monotonically decreases for an increasing uplink transmit power which suggests that downlink transmit power minimization and uplink transmit power minimization  conflict with each other. On the other hand, the objective of total harvested energy maximization does not align with the objectives of uplink and downlink transmit power minimization. It can be seen that the amount of harvested energy can only be increased by transmitting with higher uplink and/or downlink transmit power. In particular, the resource allocation policy maximizes the harvested energy using the maximum downlink and uplink transmit power allowances, which corresponds to the top corner point in Figure \ref{fig:trade_off_3D}. In fact, this is the optimal solution of single objective optimization {Problem 3} which can be found by solving (\ref{eqn:mul_obj_sdp}) with $\lambda_1=\lambda_2=0$ and $\lambda_3=1$. Similarly, the other two extreme points in the left and right corners correspond to the solutions of single-objective {Problems} 1 and {2}, which are obtained from the extreme cases of (\ref{eqn:mul_obj_sdp}) for $\lambda_1=1$ and $\lambda_2=1$, respectively.
\begin{figure}[t]
 \centering
\includegraphics[width=3.5in]{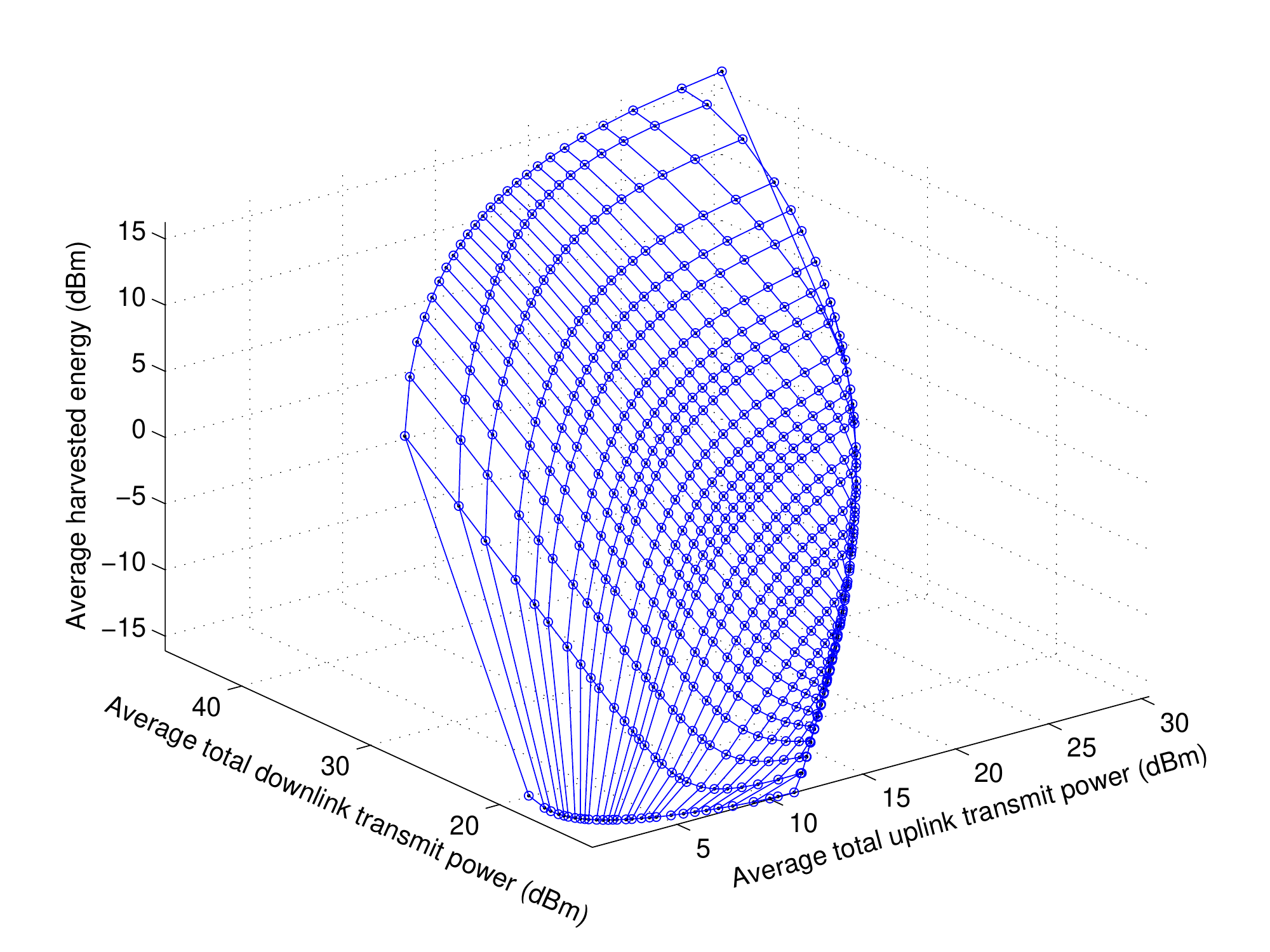}
 \caption{Trade-off region between uplink transmit power minimization, downlink transmit power minimization, and total harvested energy maximization for $N=8$. }
 \label{fig:trade_off_3D}
\end{figure}

\subsection{Average Uplink and Downlink Transmit Powers}
In Figure \ref{fig:fd_hd}, we show the trade-off between uplink and downlink transmit power minimization for different minimum required downlink SINRs, $\Gamma_{\mathrm{req},k}^\mathrm{DL}$. In particular, we select resource allocation policies with $\lambda_3=0$ and $\lambda_1+\lambda_2=1$. The points are obtained by solving (\ref{eqn:mul_obj_sdp}) for different pairs of $\lambda_1$ and $\lambda_2$. For comparison, we adopt a baseline scheme based on HD communication, where a HD radio BS is employed for transmission and reception in alternating time slots. In other words, for a given time interval, the required data rates for  uplink and downlink transmissions in each HD slot are given by $\mathrm{Rate}^\mathrm{HD-UL}_m=2\log(1+\Gamma_{\mathrm{req},m}^\mathrm{UL}), \forall m$, and $\mathrm{Rate}^\mathrm{HD-DL}_k=2\log(1+\Gamma_{\mathrm{req},k}^\mathrm{DL}), \forall k$, respectively. Thus, the required uplink and downlink SINRs in HD transmission are given by $\Gamma_{\mathrm{req},m}^\mathrm{HD-UL}=(1+\Gamma_{\mathrm{req},m}^\mathrm{UL})^2-1$ and $\Gamma_{\mathrm{req},k}^\mathrm{HD-DL}=(1+\Gamma_{\mathrm{req},k}^\mathrm{DL})^2-1$, respectively. Additionally, both SI and CCI can be avoided in the HD scenario. The baseline scheme is designed to achieve the optimal trade-off between the three considered objectives in a HD system with identical sets of weights as for the proposed FD algorithm.  In the baseline scheme, we optimize the same variables $\{\mathbf{Q},\underline{\mathbf{w}},{\cal P}\}$ and impose the same QoS requirements as in the FD case, and also apply ZF-BF detection. As shown in Figure \ref{fig:fd_hd}, significant power savings can be achieved by the proposed FD resource allocation algorithm compared to the HD system, as indicated by the double-sided arrows. Furthermore, when the downlink SINR required becomes less stringent, e.g. from $\Gamma_{\mathrm{req},k}^\mathrm{DL}=21$ dB to $\Gamma_{\mathrm{req},k}^\mathrm{DL}=15$ dB,  both the uplink and downlink transmit powers can be reduced simultaneously. This is due to the following two reasons. First, a smaller downlink transmit power is required  to satisfy the less stringent downlink SINR requirements. Second,  the decrease in  downlink transmit power reduces the self-interference impairing the uplink signal reception which improves the uplink transmit power efficiency.
\begin{figure}[t]
 \centering
\includegraphics[width=3.5in]{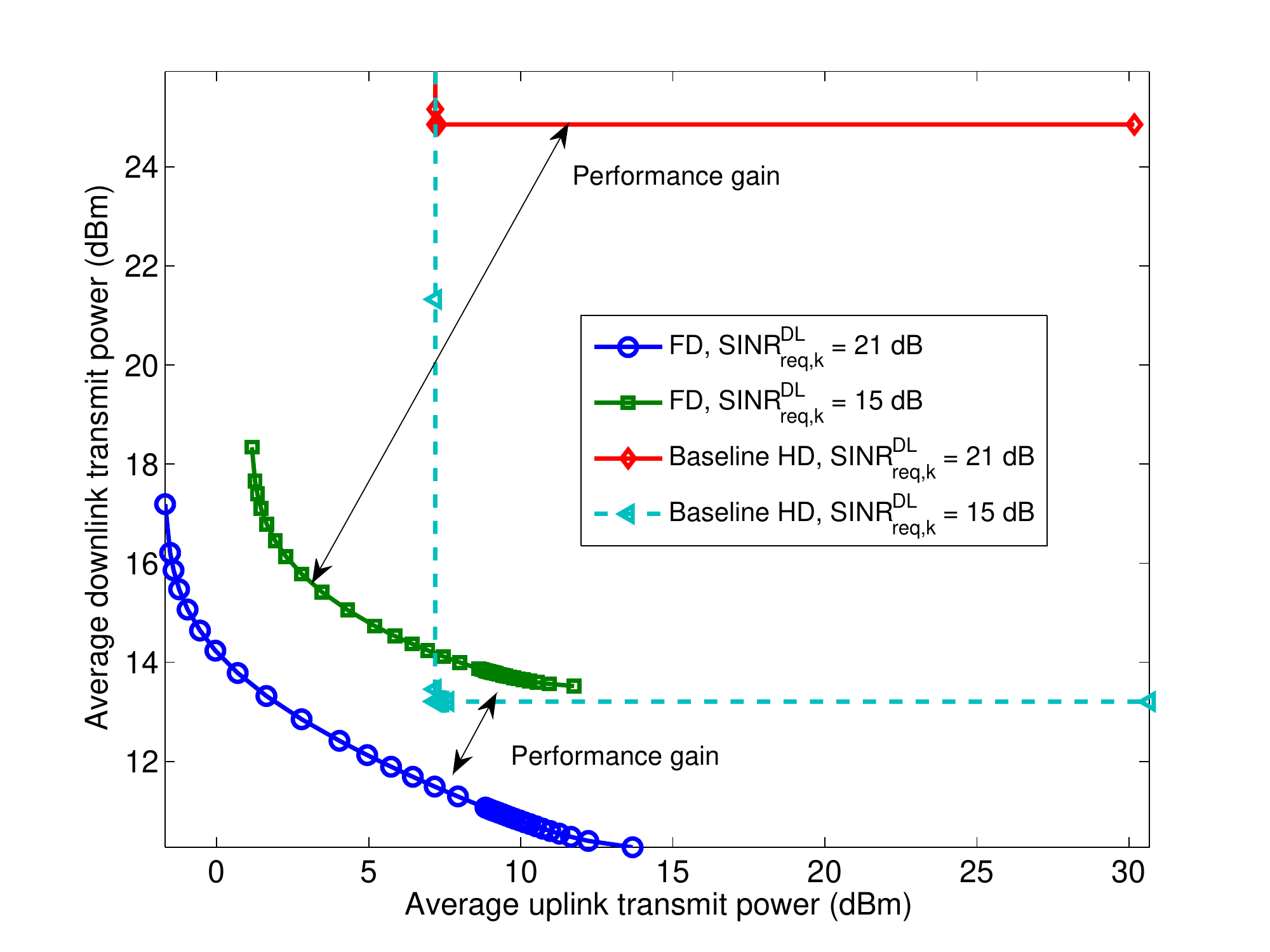}
 \caption{Average  downlink transmit power (dBm) versus average  uplink transmit power (dBm). The double-sided arrows indicate the power saving due to FD communication.}
 \label{fig:fd_hd}
\end{figure}

\subsection{Average Total Harvested Power}
In Figure \ref{fig:dleh}, we show a  trade-off between total harvested power maximization and downlink transmit power minimization. In particular, we select resource allocation policies with $\lambda_2=0$ and $\lambda_1+\lambda_3=1$. The points are obtained by solving (\ref{eqn:mul_obj_sdp}) for different pairs of $\lambda_1$ and $\lambda_3$. Besides, the HD baseline scheme is adopted again for comparison. As can be observed, the proposed FD scheme is able to provide a larger trade-off region compared to the baseline scheme. In particular, although the FD scheme suffers from self-interference, it can facilitate  power-efficient SWIPT via the proposed resource allocation optimization. Besides, a more stringent downlink minimum SINR requirement reduces the size of the trade-off region achieved by the proposed FD communication scheme. In fact, the more stringent downlink minimum SINR requirement  reduces the feasible solution set of optimization problem \eqref{eqn:mul_obj_sdp} which yields a less flexible resource allocation.
\begin{figure}[t]
 \centering
\includegraphics[width=3.5in]{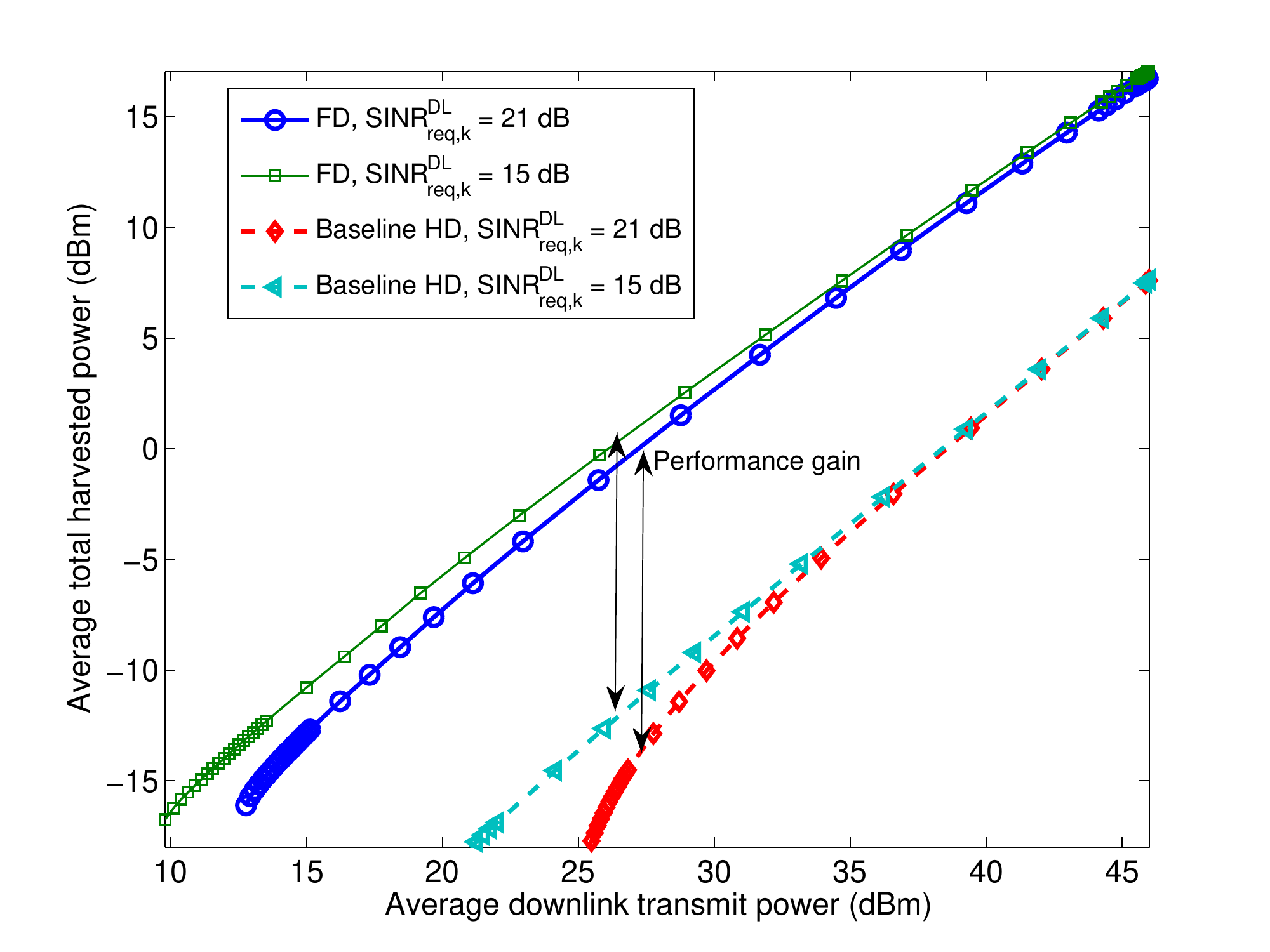}
 \caption{Average total harvested power (dBm) versus the average downlink transmit power (dBm). The double-sided arrows indicate the system performance gain
due to FD communication.}
 \label{fig:dleh}
\end{figure}

\section{Conclusion}\label{sect:conclusion}
In this paper, we designed a resource allocation algorithm for multiuser FD SWIPT systems. We proposed a MOO framework based on the weighted Tchebycheff method to study the trade-off between uplink transmit power minimization, downlink transmit power minimization, and total harvested energy maximization. The non-convex optimization problem was transformed into an equivalent rank-constrained SDP and  solved optimally by SDP relaxation. The proposed algorithm provided a set of resource allocation policies and demonstrated a remarkable performance gain in power consumption compared to a  baseline algorithm employing conventional HD transmission.

\section*{Appendix-Proof of Theorem \ref{thm:rankone}}
Theorem \ref{thm:rankone} can be proved by following a similar approach as in \cite{JR:Rui_MISO_beamforming} via investigating the Karush-Kuhn-Tucker (KKT) optimality conditions of the SDP relaxed problem (\ref{eqn:mul_obj_sdp}). The proof can be divided into two parts. In the first part, we prove that the optimal energy beamforming signal satisfies $\Rank(\mathbf{Q}^*)\leq1$. First of all, we introduce the Lagrangian of the problem as follows
\begin{eqnarray}\label{eqn:lagrangian}
&&\hspace*{-5mm}{\cal L}(\underline{\mathbf{W}},\mathbf{Q},{\cal P},\tau,\alpha,\bm{\beta},\bm{\gamma},\bm{\delta},\bm{\mu},\bm{\nu},\underline{\mathbf{X}},\mathbf{Y},\rho_1,\rho_2,\rho_3)\\
&\hspace*{-5mm}=&\hspace*{-4mm}\tau+\alpha\Big[\hspace*{-0.5mm}\Tr\big(\sum_{k=1}^K\mathbf{W}_k\hspace*{-1mm}+
\hspace*{-1mm}\mathbf{Q}\big)\hspace*{-1mm}-\hspace*{-1mm}P_{\mathrm{max}}^{\mathrm{DL}}\Big]\hspace*{-1mm}+\hspace*{-1mm}\sum_{m=1}^M\nu_m(P_m\hspace*{-1mm}-\hspace*{-1mm}P_{\mathrm{max},m}^\mathrm{UL})\notag\\
&\hspace*{-5mm}-&\hspace*{-5mm}\sum_{k=1}^K\hspace*{-0.5mm}\beta_k\hspace*{-0.5mm}\Big(\frac{\Tr(\mathbf{H}_k\hspace*{-1mm}\mathbf{W}_k)}{\Gamma_{\mathrm{req},k}^\mathrm{DL}}\hspace*{-1mm}-\hspace*{-1mm} I_k^\mathrm{DL}\hspace*{-1mm}-\hspace*{-1mm}\sigma_{\mathrm{DL},k}^2\hspace*{-0.5mm}\Big)\hspace*{-1mm}-\hspace*{-1mm}\sum_{k=1}^K\hspace*{-0.5mm}\Tr(\mathbf{X}_k\hspace*{-1mm}\mathbf{W}_k)\hspace*{-1mm}
-\hspace*{-1mm}\Tr(\mathbf{Y}\mathbf{Q})\notag\\
&\hspace*{-5mm}-&\hspace*{-5mm}\sum_{m=1}^M\hspace*{-0.5mm}\gamma_m\Big(\frac{P_m\hspace*{-0.5mm}\Tr\big(\mathbf{G}_m\mathbf{Z}_m\big)}{\Gamma_{\mathrm{req},m}^\mathrm{UL}}\hspace*{-1mm}-\hspace*{-1mm} I_m^\mathrm{UL}\hspace*{-1mm}-\hspace*{-1mm}\sigma_\mathrm{UL}^2\hspace*{-0.5mm}\Tr(\mathbf{Z}_m)\Big)\hspace*{-1mm}-\hspace*{-1mm}\sum_{m=1}^M\mu_mP_m\notag\\
&\hspace*{-5mm}-&\hspace*{-5mm}\sum_{j=1}^J\delta_j\Big(\overline{P_j^\mathrm{EH}}\hspace*{-1mm}-\hspace*{-1mm}P_{\mathrm{min},j}\Big)\hspace*{-1mm}
+\hspace*{-0.5mm}\rho_1\Big[\lambda_1\Big(\Tr\hspace*{-0.5mm}\big(\hspace*{-0.5mm}\sum_{k=1}^K\hspace*{-0.5mm}\mathbf{W}_k\hspace*{-1mm}
+\hspace*{-1mm}\mathbf{Q}\big)\hspace*{-1mm}-\hspace*{-1mm}F_1^*\Big)\hspace*{-1mm}-\hspace*{-0.5mm}\tau\Big]\notag\\
&\hspace*{-5mm}+&\hspace*{-4.5mm}\rho_2\Big[\lambda_2\Big(\hspace*{-0.5mm}\sum_{m=1}^M\hspace*{-0.5mm}P_m\hspace*{-1mm}-\hspace*{-1mm}F_2^*\Big)\hspace*{-1mm}
-\hspace*{-1mm}\tau\Big]\hspace*{-1mm}+\hspace*{-0.5mm}\rho_3\Big[\lambda_3\Big(\hspace*{-0.5mm}\sum_{j=1}^J\hspace*{-0.5mm}\overline{P_j^\mathrm{EH}}\hspace*{-1mm}-\hspace*{-1mm}F_3^*\Big)\hspace*{-1mm}-\hspace*{-0.5mm}\tau\Big]\notag,
\end{eqnarray}
where $\alpha,\bm{\beta},\bm{\gamma},\bm{\delta},\bm{\mu},\bm{\nu},\underline{\mathbf{X}},\mathbf{Y},\rho_1,\rho_2,$ and $\rho_3$ are dual variables corresponding to the associated constraints.  ${\beta_k},{\gamma_m},{\delta_j},{\mu_m},$ and ${\nu}_m$ are the elements of  dual variables $\bm{\beta},\bm{\gamma},\bm{\delta},\bm{\mu},$ and $\bm{\nu}$, respectively. Since the SDP relaxed problem satisfies Slater's constraint qualification and is convex with respect to the optimization variables, strong duality holds. Denote the optimal primal solution as $\{\underline{\mathbf{W}}^*,\mathbf{Q}^*,{\cal P}^*\}$, and the optimal dual variables as $\{\alpha^*,\bm{\beta}^*,\bm{\gamma}^*,\bm{\delta}^*,\bm{\mu}^*,\bm{\nu}^*,\underline{\mathbf{X}}^*,\mathbf{Y}^*,\rho_1^*,\rho_2^*,\rho_3^*\}$. Then, the KKT conditions used for the proof are given by:
\begin{eqnarray}
\hspace*{-6mm}\mathbf{Y}^*&\hspace*{-3mm}=&\hspace*{-2.5mm}(\alpha^*+\rho_1^*\lambda_1)\mathbf{I}-\mathbf{V},\quad\text{where}\\
\hspace*{-6mm}\mathbf{V}&\hspace*{-3mm}=&\hspace*{-2.5mm}\eta\sum_{j=1}^J(\delta_j^*\hspace*{-1mm}+
\hspace*{-1mm}\rho_3^*\lambda_3)\bm{\Omega}_j\bm{\Omega}_j^H\notag\\
&\hspace*{-3mm}-&\hspace*{-2.5mm}\sum_{m=1}^M\hspace*{-1mm}\gamma_m^*\mathbf{H}_\mathrm{SI}^H\diag(\mathbf{Z}_m)\mathbf{H}_\mathrm{SI},\\
\hspace*{-6mm}\mathbf{X}_k^*&\hspace*{-3mm}=&\hspace*{-2.5mm}\mathbf{Y}^*+\sum_{i\neq k}^K\beta_i^*\mathbf{H}_i-\frac{\beta_k^*}{\Gamma^\mathrm{DL}_{\mathrm{req},k}}\mathbf{H}_k,\,\,\forall k,\label{eqn:X_lagrangian}\\
\hspace*{-6mm}\mathbf{Y}^*\mathbf{Q}^*&\hspace*{-3mm}=&\hspace*{-2.5mm}\mathbf{0},\quad\mathbf{X}_k^*\mathbf{W}_k^*=\mathbf{0},\quad\forall k.\label{eqn:slackness_cond}
\end{eqnarray}
Since we have for the Lagrangian multiplier $\mathbf{Y}^*\succeq\mathbf{0}$, inequality $\alpha^*+\rho_1^*\lambda_1\ge\xi_\mathrm{max}$ must hold, where $\xi_\mathrm{max}$ is the largest eigenvalue of $\mathbf{V}$. If $\alpha^*+\rho_1^*\lambda_1=\xi_\mathrm{max}$, then $\Rank(\mathbf{Y}^*)=N-1$. According to the complementary slackness condition in (\ref{eqn:slackness_cond}), $\mathbf{Q}^*$ lies in the null space spanned by the column vectors of $\mathbf{Y}^*$. Thus, $\Rank(\mathbf{Q}^*)\leq1$. On the other hand, when $\alpha^*+\rho_1^*\lambda_1>\xi_\mathrm{max}$ holds, we have $\Rank(\mathbf{Y}^*)=N$ and $\mathbf{Q}^*=\mathbf{0}$. As a result, $\Rank(\mathbf{Q}^*)\leq1$ must be satisfied. In other words, at most one energy beam is needed to achieve the system design objectives.

Next, we prove the second part, i.e., $\Rank(\mathbf{W}_k^*)=1,\forall k$. It can be verified that $\beta_k^*>0$  for $\Gamma_{\mathrm{req},k}^\mathrm{DL}>0$. Besides, as proved in the first part, we have $\Rank(\mathbf{Y}^*)\ge N-1$. Since all channel variables in the system are statistically independent, $\mathbf{Y}^*$ and $\sum_{i\neq k}^K\beta_i^*\mathbf{H}_i$ span the whole signal space leading to $\Rank(\mathbf{Y}^*+\sum_{i\neq k}^K\beta_i^*\mathbf{H}_i)=N$. Then, based on the basic property of the rank of matrices, we obtain
\begin{eqnarray}
&&\hspace*{-5mm}\Rank(\mathbf{X}_k^*)\hspace*{-0.5mm}+\hspace*{-0.5mm}\Rank(\frac{\beta_k^*}{\Gamma^\mathrm{DL}_{\mathrm{req},k}}
\mathbf{H}_k)\hspace*{-0.5mm}\ge\hspace*{-0.5mm}\Rank(\mathbf{Y}^*\hspace*{-0.5mm}+\hspace*{-0.5mm}\sum_{i\neq k}^K\beta_i^*\mathbf{H}_i)\notag\\
\Longrightarrow&&\hspace*{-5mm}\Rank(\mathbf{X}_k^*)\ge  N-1.
\end{eqnarray}
As $\mathbf{W}_k^*$ lies in the null space of $\mathbf{X}_k^*$ according to (\ref{eqn:slackness_cond}), $\Rank(\mathbf{W}_k^*)\leq1$ holds. Considering $\mathbf{W}_k^*\neq\mathbf{0}$ must hold to fulfill the downlink SINR requirement, we finally conclude that $\Rank(\mathbf{W}_k^*)=1$ holds for the optimal solution. \qed


\end{document}